\documentclass[12pt]{article}
\usepackage{amsmath}
\usepackage{amssymb}
\usepackage{amsthm}
\usepackage{fullpage}
\usepackage[utf8]{inputenc}
\usepackage{mathtools}
\usepackage{url}
\usepackage[usenames,svgnames]{xcolor}
\usepackage{cite}

\usepackage{fullpage}
\usepackage[utf8]{inputenc}
\usepackage{mathtools}
\usepackage{url}
\usepackage[usenames,svgnames]{xcolor}
\usepackage{cite}
\usepackage{enumi tem}
\usepackage{datetime}
\usepackage[titletoc,title]{appendix}
\usepackage{comment}
\usepackage{mathdots}
\usepackage{graphicx}

\usepackage[pagebackref=false]{hyperref} 
\usepackage[all]{hypcap} 
\RequirePackage{color}

\usepackage{amsmath,amssymb,amsthm,tikz}
\usetikzlibrary{decorations.pathreplacing}
\usepackage{bm}
\usepackage{caption}
\usepackage[final,color]{showkeys}
\usepackage{xr}

\usepackage{mdwlist}	

\usepackage{placeins}

\usepackage{nicefrac}	

\theoremstyle{plain}
\newtheorem{theorem}{\sc Theorem}[section]
\newtheorem{lemma}[theorem]{\sc Lemma}

\newtheorem{proposition}[theorem]{\sc Proposition}
\newtheorem{corollary}[theorem]{\sc Corollary}
\newtheorem{example}{Example}[section]
\newtheorem{examples}{Example}[subsection]

\newtheorem{remark}{Remark}[section]

\theoremstyle{definition}
\newtheorem{definition}{Definition}[section]

\numberwithin{equation}{section} 

\numberwithin{equation}{section}
\theoremstyle{remark}

\newcommand{\pmeas}[2]{\xi_{#2}^{(#1)}}
\newcommand{\pmppre}{\wt{\xi}^{(n,d)}_{E'(q)}}
\newcommand{\pmp}{\pmeas{d}{E'(q)}}
\newcommand{\pme}{\pmeas{d}{E(q)}}
\newcommand{\pmh}{\pmeas{d}{H(q)}}
\newcommand{\Pmeas}[3]{\theta_{#3}^{(#2,#1)}}
\newcommand{\Pmp}{\Pmeas{d}{n}{E'(q)}}
\newcommand{\Pme}{\Pmeas{d}{n}{E(q)}}
\newcommand{\Pmh}{\Pmeas{d}{n}{H(q)}}

\newcommand{\Pf}[3]{\wt{Z}_{#3}^{(#2,#1)}}
\newcommand{\Pfp}{\Pf{d}{n}{E'(q)}}

\newcommand{\pf}[2]{Z_{#2}^{(#1)}}
\newcommand{\pfp}{\pf{d}{E'(q)}}

\newcommand{\ip}[1]{\left\langle #1 \right\rangle}

\DeclareMathOperator{\Li}{Li}
\DeclareMathOperator{\aut}{aut}

\def\ra{{\rightarrow}}

\definecolor{my-blue}{rgb}{0.0,0.0,0.6}
\definecolor{my-red}{rgb}{0.5,0.0,0.0}
\definecolor{my-green}{rgb}{0.0,0.5,0.0}
\definecolor{nicos-red}{rgb}{0.75,0.0,0.0}
\definecolor{light-gray}{gray}{0.6}
\definecolor{really-light-gray}{gray}{0.8}

\newcommand{\x}{\mathcal {X}}

\def\be{\begin{equation}}
\def\ee{\end{equation}}

\def\bea{\begin{eqnarray}}
\def\eea{\end{eqnarray}}

\def\bt{\begin{theorem}}
\def\et{\end{theorem}}

\def\bex{\begin{example}\small \rm}
\def\eex{\end{example}}

\def\bexs{\begin{examples}\small \rm}
\def\eexs{\end{examples}}

\def\ra{\rightarrow}

\def\deq{\coloneqq}

\def\br{\begin{remark}\small \rm}
\def\er{\end{remark}}

\def\&{&{\hskip -20pt}}

\def\JJ{\mathcal{J}}

\def\Ib{\mathbf{I}}

\def\Pb{\mathbf{P}}

\def\Nbb{\mathbb{N}}

\def\Rbb{\mathbb{R}}

\def\Zbb{\mathbb{Z}}

\def\Nbb{\mathbb{N}}
\def\Rbb{\mathbb{R}}
\def\Zbb{\mathbb{Z}}

\def\ep{\epsilon}

\let\Oldsection\section
\renewcommand{\section}{\FloatBarrier\Oldsection}

\let\Oldsubsection\subsection
\renewcommand{\subsection}{\FloatBarrier\Oldsubsection}

\let\Oldsubsubsection\subsubsection
\renewcommand{\subsubsection}{\FloatBarrier\Oldsubsubsection}

\newcommand{\wt}[1]{\widetilde{#1}}

\def\bN{\mathbb{N}}
\def\mP{\mathcal{P}}

\newcommand{\rb}[1]{\left(#1\right)}
\newcommand{\ab}[1]{\left[#1\right]}
\newcommand{\abs}[1]{\left|#1\right|}

\newcommand{\set}[1]{\left\{#1\right\}}

\newcommand{\lambdamap}{\Lambda^{(n)}_d}

\newcommand{\y}{\mathcal{Y}}

\newcommand{\rr}{\mathbb{R}}

\definecolor{darkgreen}{rgb}{0.0,0.5,0.0}
\definecolor{darkblue}{rgb}{0.0,0.0,0.3}
\definecolor{nicosred}{rgb}{0.65,0.1,0.1}
\definecolor{light-gray}{gray}{0.7}
\allowdisplaybreaks[1]

\newcommand{\fM}{\mathfrak{M}}

\begin{document}
\baselineskip 16pt
\begin{flushright}
CRM 3357 (2016)
\end{flushright}
\medskip
\begin{center}
\begin{Large}\fontfamily{cmss}
\fontsize{17pt}{27pt}
\selectfont
\textbf{Semiclassical asymptotics of \\ quantum weighted Hurwitz numbers}
\end{Large}\\
\bigskip
\begin{large}  {J. Harnad}$^{1,2}$ and {Janosch Ortmann}$^{1,2}$ 
 \end{large}
\\
\bigskip
\begin{small}
$^{1}${\em Centre de recherches math\'ematiques,
Universit\'e de Montr\'eal\\ C.~P.~6128, succ. centre ville, Montr\'eal,
QC, Canada H3C 3J7 } \\
\smallskip
$^{2}${\em Department of Mathematics and
Statistics, Concordia University\\ 1455 de Maisonneuve Blvd.~W.  
Montr\'eal, QC,  Canada H3G 1M8 } 
\end{small}
\end{center}
\bigskip

\begin{abstract}
   This work concerns the semiclassical asymptotics of quantum weighted double
   Hurwitz numbers.  We compute the leading term of the partition function for three  versions
   of the quantum weighted Hurwitz numbers, as well as lower order semiclassical
   corrections. The classical limit $\hbar \ra 0$  is shown to reproduce the
simple Hurwitz numbers studied by Pandharipande and Okounkov \cite{Pa, Ok}.
The KP-Toda $\tau$-function serving  as generating function for the quantum Hurwitz
numbers is shown to have the one in \cite{Pa, Ok} as classical limit  and, with
suitable scaling, so do the partition function,  the weights and expectations of Hurwitz numbers.

     \end{abstract}

\tableofcontents


\section{Introduction: weighted Hurwitz numbers and their generating functions}
\label{sec:intro}

\subsection{Hurwitz numbers}
\label{subsev:hurwitz}

Multiparametric weighted  Hurwitz numbers  were 
introduced in \cite{GH1, GH2, H1, HO} as generalizations of the notion of simple
Hurwitz numbers \cite{Hu1, Hu2, Pa, Ok} and  other special cases 
 \cite{GGN, AC1, AC2, AMMN1, AMMN2, KZ, Z}  previously studied. In general,  parametric families of
 KP or $2D$ Toda $\tau$-functions of {\it hypergeometric type} \cite{KMMM, OrSc} serve as generating 
 functions for the weighted Hurwitz numbers, which appear as coefficients  in an expansion
 over the basis of power sum symmetric functions in an auxiliary set of variables.
  The weights are determined by a parametric family  of weight
 generating functions $G(z, {\bf c})$, with parameters ${\bf c}=(c_1, c_2, \dots)$ which can either be expressed as a formal  sum
 \be
 G(z) = 1 + \sum_{i=1}^\infty g_i z^i
  \label{G_weight_gen_sum}
 \ee
 or an infinite product
 \be
 G(z) = \prod_{i=1}^\infty (1 +z c_i),
 \label{G_weight_gen_prod}
 \ee
 or some limit thereof. Comparing the two formulae, $G(z)$ can be interpreted as the generating
 function for elementary symmetric functions in the variables ${\bf c} =(c_1, c_2, \dots)$.
 \be
 g_i = e_i({\bf c}).
 \ee
 
  Another  parametrization considered in \cite{GH1, GH2, H1, HO},
consists of weight generating functions of the form
\be
 \tilde{G}(z) = \prod_{i=1}^\infty (1 - z c_i)^{-1}.
  \label{tilde_G_weight_gen}
 \ee
 The corresponding power series expansions
\be
 \tilde{G}(z) = 1 + \sum_{i=1}^\infty \tilde{g}_i z^i
 \ee
can similarly be interpreted as defining the complete symmetric functions
 \be
 \tilde{g}_i = h_i({\bf c}).
 \ee
 
Hurwitz numbers $H(\mu^{(1)}, \cdots , \mu^{(k)})$ may be defined in one of two
 equivalent ways: geometrical and combinatorial. The geometrical definition is:
 \begin{definition}
 For a set of $k$ partitions $(\mu^{(1)}, \cdots , \mu^{(k)})$ of $n$, $H(\mu^{(1)}, \cdots , \mu^{(k)})$
 is the number of distinct $n$-sheeted branched coverings $\Gamma \ra \Pb^1$ of the Riemann sphere
 having $k$ branch points $(p_1, \dots , p_k)$ with ramification profiles $\{\mu_i\}_{i=1, \dots , k}$,
 divided by the order $\aut(\Gamma)$ of the automorphism group of $\Gamma$.
 \end{definition}
 The combinatorial definition is:
 \begin{definition}  $H(\mu^{(1)}, \cdots , \mu^{(k)})$ is the number of distinct factorization
 of the identity element $\Ib \in S_n$ of the symmetric group as an ordered product 
 \be
 \Ib = h_1, \dots h_k, \quad h_i \in S_n, \quad i=1, \dots , k
\ee
where $h_i$ belongs to the conjugacy class with cycle lengths equal to the parts \
of $\mu^{(i)}$, divided by $n!$.
\end{definition}
The fact that these coincide follows from the monodromy representation
of the fundamental group of the sphere minus the branch points into $S_n$, defined by lifting
any closed loop to the branched cover,  evaluating the lift of the simple loop surrounding all the branch
points,  and decomposing the homotopy class
into an ordered product of those consisting of simple loops around each successive branch point.

Let $\mP_n$ denote the set of integer partitions of $n$ and $p(n)$  its cardinality. The Frobenius-Schur formula \cite{Frob1, Frob2, Sch, LZ} expresses the Hurwitz numbers in terms of the irreducible characters 
of  $S_n$
\be
H(\mu^{(1)}, \dots, \mu^{(k)}) = \sum_{\lambda\in \mP_n} h_\lambda^{k-2} \prod_{i=1}^k  z_{\mu^{(i)}}^{-1} \chi_\lambda(\mu^{(i)}),  
\label{frob_schur}
\ee
where $\chi_\lambda(\mu^{(i)})  $ is the irreducible character of the representation
with Young symmetry class $\lambda$ evaluated on the conjugacy class with cycle lengths equal to the parts of $\mu$;
$h_\lambda$ is the product of hook lengths of the Young diagram of partition $\lambda$ and
\be
z_\mu = \prod_{i=1}^{\ell(\mu} m_i(\mu) ! i^{m_i(\mu)}
\ee
is the order of the stabilizer of any element of the conjugacy class  $\mu$,
with $m_i(\mu)$ equal to the number of times $i$ appears as a part of $\mu$. We denote the 
weight of a partition $| \mu|$ , its length $\ell(\mu)$ and define its {\it colength} as
\be
\ell^*(\mu):= |\mu| - \ell(\mu).
\ee

\subsection{Weighted Hurwitz numbers}
\label{subsec:weighted}

Following \cite{GH1, GH2, H1, HO} we define, for each positive integer $d$
and every pair of ramification profiles $(\mu, \nu)$ (i.e. partitions of $n$),
 the weighted double Hurwitz number 
   \be
H^d_G(\mu, \nu) := \sum_{k=0}^\infty \sideset{}{'}\sum_{\substack{\mu^{(1)}, \dots \mu^{(k)} \\ \sum_{i=1}^k \ell^*(\mu^{(i)})= d}}
m_\lambda ({\bf c})H(\mu^{(1)}, \dots, \mu^{(k)}, \mu, \nu) ,
\label{Hd_G}
\ee
where
\be
m_\lambda ({\bf c}) :=
\frac{1}{\abs{\aut(\lambda)}} \sum_{\sigma\in S_k} \sum_{1 \le i_1 < \cdots < i_k}
 c_{i_\sigma(1)}^{\lambda_1} \cdots c_{i_\sigma(k)}^{\lambda_k},
 \label{monomial_sf}
\ee
is the monomial sum symmetric  function \cite{Mac} corresponding to a partition $\lambda$ of weight 
\be
|\lambda|= d = \sum_{i=1}^k \ell^*(\mu^{(i)})
\label{d_def}
\ee
whose parts $\{\lambda_i\}$ are the colengths $\ell^*(\mu^{(i)})$ in weakly descending order,
\be
|\aut (\lambda)| := \prod_{i=1}^{\ell(\lambda)} m(\lambda_i)!,
\ee
and $\sum'$ denotes the sum over all $k$-tuples of partitions $(\mu^{(1)}, \dots, \mu^{(k)})$ 
 satisfying condition (\ref{d_def}) other than the cycle type of the identity element.

By the Riemann-Hurwitz formula, the Euler characteristic of the covering surface is
\be
\chi = 2-2g = 2n - d.
\ee

For weight generating functions of the form (\ref{tilde_G_weight_gen}), the weighted double Hurwitz 
number is defined as:
\be
H^d_{\tilde{G}}(\mu, \nu) := \sum_{k=0}^\infty \sideset{}{'}\sum_{\substack{\mu^{(1)}, \dots \mu^{(k)} \\ \sum_{i=1}^k \ell^*(\mu^{(i)})= d}}
f_\lambda ({\bf c})H(\mu^{(1)}, \dots, \mu^{(k)}, \mu, \nu) 
\label{Hd_tilde_G}
\ee
where
\be
f_\lambda ({\bf c}) :=
\frac{(-1)^{\ell^*(\lambda)}}{\abs{\aut(\lambda)}} \sum_{\sigma\in S_k} \sum_{1 \le i_1 \le \cdots \le i_k} 
c_{i_\sigma(1)}^{\lambda_1},  \cdots c_{i_\sigma(k)}^{\lambda_k},
 \label{forgotten_sf}
\ee
is the ``forgotten'' symmetric function \cite{Mac}.

The particular case where all the $\mu_i$'s represent simple branching (i.e. where they are all 2-cycles) was
 studied in \cite{Pa, Ok} and corresponds to the exponential weight generating function
\be
G(z) = e^z =\lim_{k\ra \infty}(1 + {z /k})^k
\label{G_exp}
\ee
The evaluation of the monomial sum symmetric function in this limit is
\be
\label{eq:OkPP}
\lim_{k\ra \infty} m_\lambda\rb{\underbrace{\frac1m \dots, \frac1m}_{k\ \rm{times}}, 0, 0, \cdots} = \delta_{\lambda, (2, (1)^{n-2})},
\ee
so the weight is uniform  on all $k$-tuples $(\mu^{(1)}, \cdots, \mu^{(k)})$ of  partitions corresponding 
to simple branching
\be
\mu^{(i)} = (2, (1)^{n-2})
\ee
and vanishes on all others.  This is what we  view as the ``classical'' weighted (double) Hurwitz numbers.

\subsection{The $\tau$-function as generating function}
\label{subsec:taufn}

Choosing a small parameter $\beta$, the following double Schur function  expansion defines a $2D$-Toda $\tau $ function of hypergeometric
type (at the lattice point $N=0$).
\be
   \tau^{(G, \beta)} ({\bf t}, {\bf s})  = \sum_{\lambda} \ r^{(G, \beta)}_\lambda s_\lambda ({\bf t}) s_\lambda ({\bf s}),
   \label{tau_G}
  \ee
  where the coefficients $r^{(G, \beta)}_\lambda$  are defined in terms of the weight generating
  function $G$ by the following {\it content product} formula
  \be
r_\lambda^{(G, \beta)} :=   \prod_{(i,j)\in \lambda} G(\beta( j-i)),
\label{r_G_lambda}
\ee
The same formulae apply {\it mutatis mutandis} with  the replacement $G \ra \tilde{G}$ in the case of the second type of weight generating 
function $\tilde{G}$ defined by (\ref{tilde_G_weight_gen}).

By changing the expansion basis from Schur functions  \cite{Mac} to the power sum symmetric 
functions $p_\mu({\bf t}) p_\nu({\bf s})$ it  follows \cite{GH1, GH2, H1, HO}, 
that $ \tau^{(G, \beta)} ({\bf t}, {\bf s})) $ is interpretable as a generating function for the weighted double Hurwitz numbers $H^d_G(\mu, \nu)$.
\begin{theorem}
\label{tau_H_G_generating function}
The 2D Toda $\tau$-function $\tau^{(G, \beta)}({\bf t}, {\bf s})$
can be expressed as
\be
\tau^{(G, \beta)}({\bf t}, {\bf s}) =\sum_{d=0}^\infty \beta^d \sum_{\substack{\mu, \nu \\ |\mu|=|\nu|}} H^d_G(\mu, \nu) p_\mu({\bf t}) p_\nu({\bf s}).
\label{tau_H_G}
\ee
and the same formula holds under the replacement $G \ra \tilde{G}$.
\end{theorem}
The case of the classical weight generating function (\ref{G_exp}) gives the following
content product coefficient in the $\tau$-function expansion (\ref{tau_G})
\be
r^{(\exp, \beta)}_\lambda = e^{{\beta \over 2} \sum_{i=1}^{\ell(\lambda)}\lambda_i(\lambda_i - 2i +1)},
\ee
as in \cite{Ok}, and the generating function expansion (\ref{tau_G}) becomes
\be
\tau^{(\exp, \beta)}({\bf t}, {\bf s}) =\sum_{k=0}^\infty {\beta^d \over d!} \sum_{\substack{\mu, \nu \\ |\mu|=|\nu|}} H^d_{\exp} (\mu, \nu)
p_\mu({\bf t}) p_\nu({\bf s}), 
\ee
where
\be
H^d_{\exp} (\mu, \nu):= H((2, (1)^{n-2}) , \dots (\ d \ {\rm times}),  \dots, (2, (1)^{n-2})).
\ee

\subsection{Quantum Hurwitz numbers}
\label{subsec:quantumhurwitz}

A special case consists of pure quantum Hurwitz numbers \cite{GH2, H2}  which are obtained by choosing
 the parameters $c_i$ as
\be
c_i = q^i, \quad i=1, 2 \dots 
\label{c_iq_i_prime}
\ee
where $q$ is a real parameter between $0$ and $1$. 
\br
The parameter $q$ may be interpreted as $q = e^{-\epsilon}$ for a small parameter 
\be
\epsilon = \beta \hbar \omega_0, \quad \beta = 1/k T ,
\ee
where $\hbar \omega_0$ is the ground state energy, while the higher levels are integer multiples
proportional to the colength of the partition representing the ramification type of a branch point; i.e., the
degree of degeneration of the sheets
\be
\epsilon (\mu) = \ell^*(\mu) \epsilon_0.
\ee
\er

The corresponding weight generating function is 
\bea
G(z) = E'(q,z) &\& \deq \prod_{i=1}^\infty (1+ q^i z) = (-zq; q)_\infty:=1 + \sum_{i=0}^\infty E'_i(q) z^i,
\\
E'_i(q) &\& \deq \frac{q^{\frac{1}{2}i(i+1)}}{\prod_{j=1}^i (1-q^j)} =  \frac{q^{\frac{1}{2}i(i+1)}}{(q;q)_{i-1}} , \quad i \ge 1,
\label{E_prime_qz__def}
\eea
where
\be
(z;q)_k := \prod_{j=0}^{k-1}((1 - z q^j), \quad (z; q)_\infty := \prod_{j=0}^{\infty}(1 - z q^j)
\ee
is the quantum Pochhammer symbol. This is related to the quantum dilogarithm function by
\be
(1+z) E'(q, z) = e^{-\Li_2(q, -z)}, \quad \Li_2(q, z) \deq \sum_{k=1}^\infty \frac{z^k}{k (1- q^k)}.
\ee
We thus have 
\be
e_\lambda({\bf c}) = :E'_\lambda(q) = \prod_{i=1}^{\ell(\lambda)}\frac{q^{\frac{1}{2}\lambda_i(\lambda_i +1)}}{\prod_{j=1}^{\lambda_i} (1-q^j)} = \prod_{i=1}^{\ell(\lambda)}\frac{q^{\frac{1}{2}\lambda_i(\lambda_i +1)}}{(q;q)_{\lambda_i-1}}  .
\ee
The content product  coefficient entering in the $\tau$-function (\ref{tau_G}) for this case is
\bea
r^{(E'(q), \beta)}_j &\&= \prod_{k=1}^\infty (1+ q^k \beta j) = (-q\beta j; q)_{\infty} , \\
r^{(E'(q), \beta)}_\lambda(z) &\&= \prod_{k=1}^\infty \prod_{(i,j)\in \lambda} (1+ q^k \beta (j-i)) 
 = \prod_{(i,j)\in \lambda} (-q\beta(j-i); q)_\infty \cr
&\& = \prod_{k=1}^\infty (\beta q^k)^{\abs{\lambda}} (1/(\beta q^k))_\lambda, 
\eea
where $(x)_k$ denotes the rising Pochhammer symbol
\be
(x)_k := \prod_{j=1}^{k} (x + j - 1)
\ee
and
\be
(x)_{\lambda} := \prod_{i=1}^{\ell(\lambda)} (x-i +1)_{\lambda_i} = \prod_{i=1}^{\ell(\mu)}\prod_{j=1}^{\lambda_i}(x+j-i).
\ee

Making the substitutions (\ref{c_iq_i_prime}), the weights entering in (\ref{Hd_G}) evaluate to
\bea
\label{eq:WePrime}
W_{E'(q)} (\mu^{(1)}, \dots, \mu^{(k)}) &\& := m_\lambda (q, q^q, \dots )\cr
&\& =  {1\over  |\aut(\lambda)|} \sum_{\sigma\in S_k} \frac{q^{k \ell^*(\mu^{(\sigma(1))})} \cdots q^{\ell^*(\mu^{(\sigma(k))})}}{
(1- q^{\ell^*(\mu^{(\sigma(1))})}) \cdots (1- q^{\ell^*(\mu^{(\sigma(1))}} \cdots q^{\ell^*(\mu^{(\sigma(k))})})} \cr
&\& =  {1\over  |\aut(\lambda)|} \sum_{\sigma\in S_k} \frac{1}{
(q^{-\ell^*(\mu^{(\sigma(1))})} -1) \cdots (q^{-\ell^*(\mu^{(\sigma(1))})} \cdots q^{-\ell^*(\mu^{(\sigma(k))})}-1)}, \cr
&\&
\label{W_Eprime_q}
\eea
The (unnormalized) weighted Hurwitz numbers  therefore become
  \be
H^d_{E'(q)}(\mu, \nu) := \sum_{k=0}^\infty \sideset{}{'}\sum_{\substack{\mu^{(1)}, \dots \mu^{(k)} \\ \sum_{i=1}^k \ell^*(\mu^{(i)})= d}}
W_{E'(q)} (\mu^{(1)}, \dots, \mu^{(k)}) H(\mu^{(1)}, \dots, \mu^{(k)}, \mu, \nu). 
\label{Hd_E_prime_q}
\ee

Another variant  on the weight generating function for quantum Hurwitz numbers consists of choosing the 
parameters ${\bf c} = (c_1, c_2, \dots)$ in (\ref{G_weight_gen_prod}) to be
\be
c_i := q^{i-1},
\label{c_iq_i}
\ee
which gives
\bea
G(z) =&\&E(q, z) \deq (-qz;q)_\infty  = \sum_{i=0}^\infty E_i(q) z^i,
 \\
E_i(q)&\& \deq \frac{q^{\frac{1}{2}i(i-1)}}{(q;q)_{i-1}}, \quad i \ge 1.
\label{Eqz_def}
\eea
This is related to the quantum dilogarithm function by
\be
E(q, z) = e^{-\Li_2(q, -z)}, \quad \Li_2(q, z) \deq \sum_{k=1}^\infty \frac{z^k}{k (1- q^k)}.
\ee
We thus have 
\be
E_\lambda(q) = \prod_{i=1}^{\ell(\lambda)}\frac{q^{\frac{1}{2}\lambda_i(\lambda_i -1)}}{{(q;q)_{\lambda_i -1}}} \\
\ee
The content product  coefficient entering in the $\tau$-function (\ref{tau_G}) for this case is
\bea
r^{E(q)}_j(z) &\&= \prod_{k=0}^\infty (1+ q^k z j) = (-zj; q)_\infty, \\
r^{E(q)}_\lambda(z) &\&= \prod_{k=0}^\infty \prod_{(i,j)\in \lambda} (1+ q^k z (j-i)) 
 = \prod_{(i,j)\in \lambda} (-z(j-i); q)_\infty \cr
&\&= \prod_{k=0}^\infty (zq^k)^{\abs{\lambda}} (1/(zq^k))_\lambda
\eea

The weights entering in (\ref{Hd_G}) evaluate to
\begin{align}
	\label{eq:We}
W_{E(q)} (\mu^{(1)}, \dots, \mu^{(k)}) & \deq {1\over |\aut(\lambda)|}
\sum_{\sigma\in S_k} \sum_{0 \le i_1 < \cdots < i_k}^\infty q^{i_1 \ell^*(\mu^{(\sigma(1))})} \cdots q^{i_k \ell^*(\mu^{(\sigma(k))})} \\
& = {1\over  |\aut(\lambda)|}\sum_{\sigma\in S_k} \frac{q^{(k-1) \ell^*(\mu^{(\sigma(1))})} \cdots q^{\ell^*(\mu^{(\sigma(k-1))})}}{
(1- q^{\ell^*(\mu^{(\sigma(1))})}) \cdots (1- q^{\ell^*(\mu^{(\sigma(1))})} \cdots q^{\ell^*(\mu^{(\sigma(k))})})},
\label{W_E_q}
\end{align}
and the  weighted Hurwitz numbers  therefore become
  \be
H^d_{E(q)}(\mu, \nu) := \sum_{k=0}^\infty \sideset{}{'}\sum_{\substack{\mu^{(1)}, \dots \mu^{(k)} \\ \sum_{i=1}^k \ell^*(\mu^{(i)})= d}}
W_{E(q)} (\mu^{(1)}, \dots, \mu^{(k)}) H(\mu^{(1)}, \dots, \mu^{(k)}, \mu, \nu). 
\label{Hd_E_q}
\ee

A third variant on the  weight generating function for quantum Hurwitz numbers consists 
of choosing it of the form (\ref{tilde_G_weight_gen}) with parameters ${\bf c} = (c_1, c_2, \dots)$ again chosen as in (\ref{c_iq_i}).
This gives
\bea
\tilde{G}(z) =H(q,z)&\&\deq \prod_{k=0}^\infty (1-q^k z)^{-1} ={1\over (-z;q)_\infty}= e^{\Li_2(q, z)} = \sum_{i=0}^\infty H_i(q)z^i,
\label{GHq} \\
H_i(q) &\&\deq \frac{1}{(q;q)_{i-1}},
\quad H_\lambda(q) = \prod_{i=1}^{\ell(\lambda)}\frac{1}{(q;q)_{\lambda_i-1}} \\
H(q, \JJ) &\&= \prod_{k=0}^\infty\prod_{a=1}^n (1- q^k z\JJ_a)^{-1}, \\
r^{H(q)}_j(z) &\&= \prod_{k=0}^\infty (1- q^k z j)^{-1} = {1\over (-z;q)_\infty}, \\
r^{H(q)}_\lambda(z) &\&= \prod_{k=0}^\infty \prod_{(i,j)\in \lambda} (1- q^kz (j-i)) ^{-1} 
 = \prod_{(i,j)\in \lambda} {1\over (-z(j-i);q)_\infty} \cr
&\& = \prod_{k=0}^\infty (-1/(zq^k))^{-\abs{\lambda}} (-1/(zq^k))^{-1}_\lambda. 
\label{rHq}
\eea

The weights entering in (\ref{Hd_G}) then evaluate to
\begin{align}
\label{eq:Wh}
W_{H(q)} (\mu^{(1)}, \dots, \mu^{(k)}) & \deq
{(-1)^{\ell^*(\lambda)}\over   |\aut(\lambda)|}\sum_{\sigma\in S_k} \sum_{0 \le i_1 \le \cdots \le i_k}^\infty q^{i_1 \ell^*(\mu^{(\sigma(1))})} \cdots q^{i_k \ell^*(\mu^{(\sigma(k))})} \\
&= {(-1)^{\ell^*(\lambda)}  \over  |\aut(\lambda)|}\sum_{\sigma\in S_k} \frac{1}{
(1- q^{\ell^*(\mu^{(\sigma(1))})}) \cdots (1- q^{\ell^*(\mu^{(\sigma(1))})} \cdots q^{\ell^*(\mu^{(\sigma(k))})})}
\label{W_H_q}
\end{align}
and the weighted Hurwitz numbers  become
  \be
H^d_{H(q)}(\mu, \nu) := \sum_{k=0}^\infty \sideset{}{'}\sum_{\substack{\mu^{(1)}, \dots \mu^{(k)} \\ \sum_{i=1}^k \ell^*(\mu^{(i)})= d}}
W_{H(q)} (\mu^{(1)}, \dots, \mu^{(k)}) H(\mu^{(1)}, \dots, \mu^{(k)}, \mu, \nu). 
\label{Hd_H_q}
\ee

\subsection{Classical limit of the generating function for quantum Hurwitz numbers}
\label{classical_limit_weight_gen}
Choosing 
\be
q= e^{-\epsilon}
\ee
with $\epsilon$  a small positive number, and taking the limit $\epsilon \ra 0^+$
of the scaled quantum dilogarithm function $Li_2(q, \epsilon z)$ gives
\be
\lim_{\epsilon \ra 0^+} Li_2(q, \epsilon z) =z.
\ee
It follows that all three generating functions $E(q,z), E'(q,z)$ and $H(q,z)$ have as scaled limits
the generating function for the Okounkov-Pandharipande simple (single and double) Hurwitz numbers
\be
\lim_{\epsilon \ra 0^+} E(q,z)= \lim_{\epsilon \ra 0^+} E'(q,z)= \lim_{\epsilon \ra 0^+} H(q,z) = e^z
\ee
The corresponding scaled limit of the generating $\tau$-functions for all three versions of quantum 
weighted Hurwitz numbers therefore coincides with the generating function for simple Hurwitz
numbers considered in \cite{Pa, Ok}
\be
\lim_{\epsilon \ra 0^+} \tau^{E(q), \epsilon \beta)}({\bf t}, {\bf s}) = \lim_{\epsilon \ra 0^+} \tau^{E'(q), \epsilon \beta)}({\bf t}, {\bf s})
= \lim_{\epsilon \ra 0^+} \tau^{H(q), \epsilon \beta)}({\bf t}, {\bf s}) = \tau^{(\exp,  \beta)}({\bf t}, {\bf s}). 
\ee

Equivalently, this implies the limit
\be
\label{eq:OkPaLimit}
\lim_{\epsilon \ra 0^+} \epsilon^d H^d_{E'(q= e^{-\epsilon})} = H^d_{\exp} (\mu, \nu).
\ee
(cf. Theorem \ref{Hd_E_prime_semiclassical} and Remark \ref{rmk:OkPaLimit}.)


\section{Probabilistic approach to quantum Hurwitz numbers}

Since $W_{E'(q)}\rb{\mu^{(1)},\ldots,\mu^{(k)}}$ is always real, positive and normalizable, we can interpret $H^d_{E'(q)}$ in terms of an expectation.
For $k\in\set{1, \dots , d}$ consider the (finite) set of $k$-tuples
\begin{align}
		\label{eq:defM}
		\fM_{d,k}^{(n)} & = \set{\rb{\mu^{(1)}, \ldots, \mu^{(k)}} \in \rb{\mP_n}^k \colon \sum_{j=1}^k \ell^\ast \rb{\mu^{(j)}} =d } 
		\intertext{and their disjoint union}  
		\fM_d^{(n)}&=\coprod_{k=1}^{d} \fM^{(n)}_{d,k}.
\end{align}
Define a measure $\Pmp$ on $\fM^{(n)}_{d}$ by
\begin{align}
	\label{eq:defTheta}
	\Pmp \rb{ \rb{\mu^{(1)}, \ldots, \mu^{(k)}} } & = \frac1{\Pfp} W_{E'(q)}\rb{\mu^{(1)},\ldots,\mu^{(k)} },
	\intertext{where the \emph{partition function} $\Pfp$ is defined so that $\Pmp$ is a probability measure; that is,}
	\label{eq:defZt}
	\Pfp &= \sum_{k=1}^d \sum_{\fM^{(n)}_{d,k}} W_{E'(q)}\rb{\mu^{(1)},\ldots,\mu^{(k)} }.
	\intertext{We then have the expectation value}
	\label{eq:wHexp}
	\ip{H\rb{\cdot,\ldots,\cdot,\mu,\nu}}_{\Pmp} & = \frac1{\Pfp} H_{E'(q)}^d(\mu,\nu),
\end{align}
where $\ip{\cdot}_{\Pmp}$ denotes integration with respect to the measure $\Pmeas dnq$.
%


\begin{definition}
	For $n,d\in\Zbb_{>0}$ define the function $\lambdamap\colon \fM_d^{(n)}\longrightarrow \mP_d$ as follows: 
\begin{align}
	\lambdamap & \colon \rb{\mu^{(1)},\ldots,\mu^{(k)}}\longmapsto \lambda
	\intertext{where $\lambda$ is the unique partition of $d$ such that}
	\set{\lambda_1,\ldots,\lambda_k} & = \set{ \ell^\ast\rb{\mu^{1} },\ldots, \ell^\ast\rb{\mu^{(k)} }  }.
\end{align}
\end{definition}

The weight of the partition $\lambdamap\rb{\mu^{(1)},\ldots,\mu^{(k)}}$ is thus the sum of colengths of $\mu^{(1)},\ldots,\mu^{(k)}$. Letting $\mP_{n,k}$ denote the set of integer partitions, the image of $\fM_{d,k}^{(n)}$ under $\lambdamap$ is thus $\mP_{d,k}$.

Since $W_{E'(q)}\rb{\mu^{(1)},\ldots,\mu^{(k)}}$ depends on the partitions $\mu^{(1)},\ldots,\mu^{(k)}$ only through their colength it makes sense to consider the push-forward
\begin{align}
	\label{eq:PushForward}
	\pmppre & = \rb{\lambdamap}_\ast\Pmp
	\intertext{of $\Pmp$ under $\lambdamap$ (as a measure on $\mP_d$). Let $p(n,k):=\abs{\mP_{n,k}}$ denote the cardinality of $\mP_{n,k}$ and observe that, for any $\lambda\in \mP_d$,}
	\label{eq:Cardinality}
	\abs{\rb{\lambdamap}^{-1}(\lambda)}& = \prod_{j=1}^{\ell(\lambda)} p\rb{n,n-\lambda_j}.
	\intertext{Therefore}
	\label{eq:defXiPrel}
	\pmppre(\lambda) & = \frac1{\Pfp}\, \rb{\prod_{j=1}^{\ell(\lambda)} p\rb{n,n-\lambda_j} } w_{E'(q)}(\lambda)
\end{align}
where $w_{E'(q)}$ is defined as the weight function $w_{E'(q)}\colon \mP_d\longrightarrow [0,\infty)$ satisfying
\begin{align}
	\label{eq:defw}
	w_{E'(q)}(\lambda)  & = \frac{\Phi_{E'(q)}(\lambda_1,\ldots,\lambda_{\ell(\lambda)})}{\abs{\aut(\lambda)}}
	\intertext{with $\Phi\colon\coprod_{m\in\Nbb} \Rbb^m\longrightarrow \rr$ defined by}
	\label{eq:defPhi}
	\Phi_{E'(q)}\rb{x_1,\ldots,x_m} & = \sum_{\sigma \in S_m} \prod_{j=1}^m \rb{ q^{-\sum_{i=1}^j x_{\sigma(i)} } -1 }^{-1}.
\end{align}

\begin{lemma}
	\label{lem:countColengths}
	For any $n,\ell\in\bN$ with $n\geq 2\ell$ we have
	\begin{align}
		p(n,n-\ell)=p(\ell).
	\end{align}
\end{lemma}

\par\noindent The proof of this lemma is given in Section \ref{sec:proofs}. From now on we always assume that $n\geq 2d$. We  also denote
\begin{align}
	p(\lambda)=\prod_{j=1}^{\ell(\lambda)} p\rb{\lambda_j}.
\end{align}
From the above discussion and Lemma \ref{lem:countColengths} we have the following result. For $d\in\Zbb_{\geq 0}$ and $q\in (0,1)$ let
\begin{align}
	\label{eq:defpfn}
	\pfp&:= \sum_{\lambda\in\mP_d} p(\lambda)\, w_{E'(q)}(\lambda)
	\intertext{and define a probability measure on $\mP_d$ by}
	\label{eq:defXi}
	\pmp(\lambda) &:= \frac1{\pfp}\, p\rb{\lambda} w_{E'(q)}(\lambda)  \quad\quad\forall\, \lambda\in\mP_d.
\end{align}

\begin{proposition}
	Let $n,d\in\Zbb_{>0}$ with $n\geq 2d$. Then
	\begin{enumerate}
		\item The partition function $\Pfp$ does not depend on $n$: 
		\begin{align}
			\Pfp & = \pfp
		\end{align}
		\item The probability measure $\pmppre$ does not depend on $n$: for any $\lambda\in\mP_d$,
		\begin{align}
			\label{eq:pmppre}
			\pmppre (\lambda) &= \pmp (\lambda)
		\end{align}
	\end{enumerate}
\end{proposition}

We conclude this section by explaining how this extends to the other two quantum weight generating functions $E(q)$ and $H(q)$.

\begin{definition}
	Define probability measures $\Pme$ and $\Pmh$ on $\fM_d^{(n)}$ as in \eqref{eq:defTheta} and \eqref{eq:defZt}, replacing $W_{E'(q)}$ by $W_{E(q)}$ and $W_{H(q)}$ respectively, whenever it occurs.
\end{definition}

Equations \eqref{eq:wHexp}, \eqref{eq:PushForward}--\eqref{eq:defPhi} and \eqref{eq:defpfn}--\eqref{eq:pmppre}  apply mutatis mutandis, replacing $E'(q)$ by $E(q)$ and $H(q)$ respectively.

\section{Classical limits and asymptotic expansion}

\label{sec:scl}

In this section we state our asymptotic results for $q\longrightarrow 1^-$; all proofs are given in the following section. 


\subsection{Classical limit}

 We begin by stating the classical limits.
 \begin{definition}
	The \emph{Dirac measure} $\delta_x$ at $x\in S$ on a measurable space $(S,\Sigma)$ is defined by
	\begin{align}
		\label{eq:defDirac}
		\delta_x(A) & = \begin{cases}
			1\quad&\text{if } x\in A\\
			0& \text{otherwise}
		\end{cases}
	\end{align}
	for all $A\in\Sigma$.
 \end{definition}

Recall that $m_k(\lambda)$ denotes the number of blocks of size $k$ in a partition $\lambda$. We will use the following notation:

\begin{definition}
	For $\lambda\in\mP_d$,
	\label{def:notationP}
	\begin{align}
		\lambda & = \rb{1^{m_1(\lambda)},2^{m_2(\lambda)},\ldots}.
		\intertext{We will also use the following notation for partitions with at most two different part lengths, namely $\ell\in\Zbb_{>0}$ and 1: we write }
		\label{eq:specialPart}
		\bm{\ell}^m_n=\rb{1^{n-m},\ell^m},
	\end{align}
	When the weight of the partition is clear from context we simply write $\bm{\ell}^m: = \bm{\ell}^m_n$. When $m=1$ we write $\bm{\ell}_n:= \bm{\ell}^1_n$, or simply $\bm{\ell}$. 
\end{definition}

\begin{figure}[ht]

	\begin{center}
	\begin{tikzpicture}[>=latex,scale=0.5]

		\draw (0,0)--(1,0)--(1,1)--(0,1)--(0,0);
		
		\foreach \y in {1,2,3,4}{
			\draw (0,\y) -- +(0,1) -- +(1,1) -- +(1,0) -- +(0,0);
		}
		
		\draw [decorate,decoration={brace,amplitude=10pt},xshift=-4pt,yshift=0pt]
		(0,5) -- (0,8) node [black,midway,xshift=-0.6cm] 
		{\footnotesize $m$};
		
		\draw [decorate,decoration={brace,amplitude=10pt},xshift=0pt,yshift=4pt]
		(0,8) -- (5,8) node [black,midway,yshift=0.6cm] 
		{\footnotesize $\ell$};
		
		\foreach \y in {5,6,7}{
			\foreach \x in {0,1,2,3,4}{
				\draw (\x,\y) -- +(0,1) -- +(1,1) -- +(1,0) -- +(0,0);
			}
			
		}

		\foreach \y in {0,1,2,3,4,5,6,7}{
			\draw (10,\y) -- +(0,1) -- +(1,1) -- +(1,0) -- +(0,0);
		}

		\foreach \x in {11,12,13,14,15,16}{
			\draw (\x,7) -- +(0,1) -- +(1,1) -- +(1,0) -- +(0,0);
		}

		\draw [decorate,decoration={brace,amplitude=10pt},xshift=0pt,yshift=4pt]
		(10,8) -- (17,8) node [black,midway,yshift=0.6cm] 
		{\footnotesize $\ell$};
													
	\end{tikzpicture}
	\end{center}
	\footnotesize{\caption{The partitions $\bm{\ell}^m_n=\bm{5}^3_{20}$ (left) and $\bm\ell=\bm{7}$ (with $m=1$ and $n=14$ suppressed from the notation, right)}}
\label{fig:specialPart}
\end{figure}
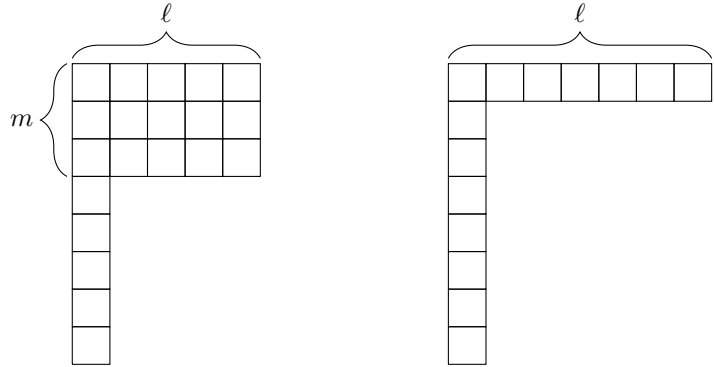

\begin{theorem}
	\label{thm:Downstairs}
	Let $d\in\Zbb_{>0}$. As $q\longrightarrow 1^-$, each of the sequence of measures $\rb{\pmp}_{q<1}$, $\rb{\pme}_{q<1}$ and $\rb{\pmh}_{q<1}$ on $\mP_d$ converges weakly to the Dirac measure $\delta_{(1^d)}$ at $(1^d)\in\mP_d$.
\end{theorem}

By the discussion in Section 2 this translates to a convergence result on $\fM^{(n)}_d$:

\begin{corollary}
	\label{cor:Upstairs}
	 If $d\geq 2n$ then each of the sequence of measures $\Pmp$, $\Pme$ and $\Pmh$ on $\fM_d^{(n)}$ converges weakly, as $q\longrightarrow 1^-$, to the Dirac measure at $(\underbrace{\bm{2},\ldots,\bm{2} }_{d\text{ terms}})$ (in the notation of \eqref{eq:specialPart})
\end{corollary}

\begin{remark}
	Observe that the limiting measure in Corollary \ref{cor:Upstairs} corresponds to the Okounkov / Pandharipande measure from \eqref{eq:OkPP}.
\end{remark}


\subsection{Semiclassical corrections}

We now turn to semiclassical asymptotics.  Throughout we set $q=e^{-\ep}$ and let $\ep\longrightarrow 0^+$. We begin by giving the asymptotic expansion for each weight. For any $\lambda\in\mP_d$ define 
\begin{align}
	w_0(\lambda) & = \sum_{\sigma\in S_{\ell(\lambda)}}  \frac1{ \prod_{j=1}^{\ell(\lambda)} \sum_{i=1}^j \lambda_{\sigma(i)} }\\
	w_1(\lambda) & = \frac12 \sum_{\sigma\in S_{\ell(\lambda)}}\sum_{r=1}^{\ell(\lambda)} \frac{\sum_{i=1}^r \lambda_{\sigma(i)} }{\prod_{j=1}^{\ell(\lambda)} \sum_{i=1}^j \lambda_{\sigma(i)} }
\end{align}

\begin{theorem}
	\label{prop:weightSC}
	For any $\lambda\in\mP_d$ we have
	\begin{align}
		\ep^{-\ell(\lambda)} w_{E'(e^{-\ep})}(\lambda) & = w_0(\lambda)   + \ep w_1(\lambda) \,      +O\rb{\ep^{2}}\\
		\ep^{-\ell(\lambda)} w_{E(e^{-\ep})}(\lambda) & = w_0(\lambda) + \ep\rb{w_1(\lambda) - dw_0(\lambda)}+O\rb{\ep^{2}}\\
		\ep^{-\ell(\lambda)} w_{H(e^{-\ep})}(\lambda) & = w_0(\lambda) + \ep\rb{w_1(\lambda) - \frac{\ell(\lambda)(\ell(\lambda)+1)}2\, d w_0(\lambda)} +O\rb{\ep^{2}}.
	\end{align}
\end{theorem}

From this result one can deduce the following semiclassical expansion for the partition function:

\begin{theorem}
	\label{cor:pfSC}
	For $d\in\Zbb_{\geq 0}$ and $q=e^{-\ep}$ we have
	\begin{align}
		\ep^d\, \pf d{E'(e^{-\ep})}	&= \frac{1}{d!} +\ep\, \frac{3-d}{4(d-1)!} + O\rb{\ep^{2}},\\
		\ep^d\, \pf d{E(e^{-\ep})}	&= \frac{1}{d!} +\ep\, \frac{5+d}{4(d-1)!} + O\rb{\ep^{2}},\\
		\ep^d\, \pf d{H(e^{-\ep})}	&= \frac{1}{d!} +\ep\, \frac{d+1}{(d-1)!} + O\rb{\ep^{2}}.
	\end{align}
\end{theorem}

We also obtain a convergence result for the weighted Hurwitz numbers. (Recall our notation for partitions from Definition \ref{def:notationP}.)

\begin{theorem}
\label{Hd_E_prime_semiclassical}
	For any $\mu,\nu\in\mP_n$ and $G\in\set{E',E,H}$ we have
	\label{thm:WeightedHurwitz}
	\begin{align}
		\ep^d\ H_{G(q)}^{d}(\mu,\nu) & = \frac{1}{d!}\, H(\underbrace{\bm{2},\ldots,\bm{2}}_{d\text{ times}},\mu,\nu) \\
		&\quad + \ep \left[ \gamma(1) H(\underbrace{\bm{2},\ldots,\bm{2}}_{d-1\text{ times}},\bm{3},\mu,\nu)) + \gamma(1) H(\underbrace{\bm{2},\ldots,\bm{2}}_{d-1\text{ times}},\bm{2^2},\mu,\nu)) \right.\\
		 &\quad\quad\left.+\, \gamma_{G(q)}(2) H(\underbrace{\bm{2},\ldots,\bm{2}}_{d\text{ times}},\mu,\nu)\right] + O\rb{\ep^{2}}.
	\end{align}
	where $\gamma(1)= \frac1{(d-1)!}$ and
	\begin{align}
		\gamma_{E'(q)}(2) = -\frac{d+1}{4(d-1)!},\quad\quad \gamma_{E(q)}(2) =-\frac{3-d}{4(d-1)!},\quad\quad \gamma_{H(q)}(2) = \frac{d+1}{2(d-1)!}.
	\end{align}
\end{theorem}

\begin{remark}
	\label{rmk:OkPaLimit}
	In particular, for $G\in\set{E',E,H}$, we have
	\begin{align}
		\lim_{\ep\to 0} H_{G(e^{-\ep})}\rb{\mu,\nu} & = \frac1{d!}\, H_{\exp}(\mu,\nu),
	\end{align}
	which includes \eqref{eq:OkPaLimit} for $G=E'$.
\end{remark}

\section{Proofs}
\label{sec:proofs}

\begin{proof}[Proof of Lemma \ref{lem:countColengths}]
		Consider the function $f\colon\mP_{n,n-\ell}\longrightarrow \mP_\ell$ defined as follows. Let $\lambda\in\mP_{n,n-\ell}$, then the first column of the Young diagram of $\lambda$ has $n-\ell$ boxes. Remove these to obtain a partition $\nu:=f(\lambda)$ of $\ell$. This function has an inverse: for $\nu\in\mP_\ell$ simply add a new column with $n-\ell$ to the left of the Young diagram of $\nu$. Since $n-\ell\geq \ell$ by assumption the result is the Young diagram of an integer partition $\lambda:=f^{-1}(\nu)$: it is easy to see that $\lambda\in\mP_n$ and that $\ell(\lambda)=n-\ell$.
	\end{proof}

We only detail the proofs of the results from Section \ref{sec:scl} for the case $E'(q)$. The corresponding results for $E(q)$ and $H(q)$ follow analogously, using \eqref{W_E_q} and \eqref{W_H_q}. The proofs all rely on the following asymptotic expansion of $\Phi_{E'(e^{-\ep})}$ as $\ep\longrightarrow 0$:

\begin{lemma}
	\label{lem:fundamental}
	let $x_1,\ldots, x_m\in\Zbb_{>0}$. Then, as $\ep\longrightarrow 0$
	\begin{align}
		\Phi_{E'(e^{-\ep})}\rb{x_1,\ldots,x_m} &= \ep^{-m} \sum_{\sigma\in S_m} \ab{\frac{ 1  }{ \prod_{j=1}^m \sum_{i=1}^j x_{\sigma(i)} }   - \frac{ \ep }2 \, \sum_{r=1}^m \frac{\sum_{i=1}^r x_{\sigma(i)}}{\prod_{j=1}^m \sum_{i=1}^j x_{\sigma(i)} }}  +O\rb{\ep^{2-m}} 
	\end{align}
\end{lemma}

\begin{proof}
	A direct computation yields
	\begin{align}
		\Phi_{E'(e^{-\ep})}\rb{x_1,\ldots,x_m} & = \sum_{\sigma\in S_m} \prod_{j=1}^m \rb{e^{\epsilon\sum_{i=1}^j x_{\sigma(i)}}-1}^{-1}\\
		& = \sum_{\sigma\in S_m} \prod_{j=1}^m \frac{\ep^{-1}}{\sum_{i=1}^j x_{\sigma(i)}} \rb{1 + \frac{\ep}2 \sum_{i=1}^j x_{\sigma(i)} + O(\ep^2)}^{-1}\\
		& = \ep^{-m}\, \sum_{\sigma\in S_m} \prod_{j=1}^m \rb{ \frac{1}{\sum_{i=1}^j x_{\sigma(i)}}  - \frac{\ep}2  + O(\ep^2)}\\
		& = \ep^{-m}\, \sum_{\sigma\in S_m} \rb{ \prod_{j=1}^m \frac{1}{\sum_{i=1}^j x_{\sigma(i)}} - \frac{\ep}2 \sum_{j=1}^m \frac{\sum_{i=1}^j x_{\sigma(i)}}{\prod_{r=1}^m\sum_{i=1}^r x_{\sigma(i)}}  + O(\ep^2)}
	\end{align}
	as claimed.
\end{proof}

\par\noindent By considering the highest order terms it follows immediately that,  letting 
	\begin{align}
		d=\sum_{r=1}^m x_r&\geq m,
		\intertext{we have}
		\lim_{\ep\to 0} \ep^d \Phi_{E'(e^{-\ep})} \rb{x_1,\ldots,x_m} & = \begin{cases}
			\frac1{d!} \quad\quad & \text{if } d=m\\
			0 & \text{if } d>m.
		\end{cases}
	\end{align}

\par\noindent This completes the proof of Theorem \ref{thm:Downstairs} and hence also Corollary \ref{cor:Upstairs}. Setting $q=e^{-\ep}$ and considering additionally the terms of order $\ep^{1-d}$ gives Theorem \ref{prop:weightSC}.


\par\noindent Moreover we obtain the following intermediate result:

\begin{proposition}
	\label{lem:f}
	For any function $f\colon \fM_{d}^{(n)}\longrightarrow\Rbb$,
	\begin{align}
		\label{eq:f}
		\ep^{d} \sum_{\fM_d^{(n)}} &f\rb{\mu^{(1)},\ldots,\mu^{(k)}} W_{E'(q)}\rb{\mu^{(1)},\ldots,\mu^{(k)}} =\frac{1}{d!}  f(\underbrace{\bm2,\ldots,\bm2}_{\text{d }times}) \\
		& \quad + \frac{\ep}{(d-1)!}\,\ab{f\rb{ \underbrace{\bm 2,\ldots,\bm 2}_{d-1\text{ times}}, \bm 3 } + f\rb{ \underbrace{\bm 2,\ldots,\bm 2}_{d-1\text{ times}} , \bm{2^2}} - \,\frac{d+1}{4} f(\underbrace{\bm2,\ldots,\bm2}_{\text{d }times})} \\
		& \quad + O\rb{\ep^{2}}
	\end{align}
	where we recall that $\bm 2 = (1^{n-1},2)$ and $\bm 3 =( 1^{n-3},3)$ and $\bm {2^2} = (1^{n-4},2^2)$. 
\end{proposition}

\begin{proof}
	From Lemma \ref{lem:fundamental} it follows that $w_{E'(q)}(\lambda)$ contributes terms of order $\ep^{-\ell(\lambda)}$ and lower. Thus the only terms in \eqref{eq:f} that are not $o(\ep^{-d+1})$ correspond to elements $(\mu^{(1)},\ldots,\mu^{(k)})$ of $\fM_d^{(n)}$ such that $\lambda=\lambdamap(\mu^{(1)},\ldots,\mu^{(k)})$ has length $d$ or $d-1$, i.e. $\lambda\in\set{\bm{1},\bm{2}}$. (recalling once more the notation from Definition \ref{def:notationP}). Therefore,
	\begin{align}
		\label{eq:first}
		\sum_{\fM_d^{(n)}} &f\rb{\mu^{(1)},\ldots,\mu^{(k)}} W_{E'(q)}\rb{\mu^{(1)},\ldots,\mu^{(k)}} = \frac{p(\bm 1)}{\abs{\aut(\bm 1)}} \Phi_{e^{-\ep}}(1,\ldots,1) \sum_{\Lambda_d^{-1}(1^d)} f\rb{\mu^{(1)},\ldots,\mu^{(k)}} \\
		\label{eq:second}
		& \quad + \frac{p(\bm 2)}{\abs{\aut(\bm 2)}} \Phi_{e^{-\ep}}(2,1,\ldots,1) \sum_{\Lambda_d^{-1}(\bm2)} f\rb{\mu^{(1)},\ldots,\mu^{(k)}}  + O(\ep^{2-d}).
	\end{align}
	We first deal with the term in \eqref{eq:first}: $p(\bm{1})=1$ and $\aut(\bm{1})=d!$. Further, by Lemma \ref{lem:fundamental},
	\begin{align}
		\label{eq:413}
		\Phi_{E'(q)} (\underbrace{1,\ldots,1}_{d\text{ times}}) & = \ep^{-d} \sum_{\sigma\in S_d} \rb{\frac1{d!} - \frac\ep2 \frac{d(d+1)/2}{d!} +O\rb{\ep^2} } = \ep^{-d} \rb{1- \ep \,\frac{d(d+1)}{4 }} + O\rb{\ep^{-d+2}}.
	\end{align}
	For the terms in \eqref{eq:second}: $p(\bm{2})=p(2)=2$ and $\aut(\bm{2})=(d-1)!$. This time we only need the first order approximation of Lemma \ref{lem:fundamental}, and we obtain
	\begin{align}
		\Phi_{E'(e^{-\ep})} (2,\underbrace{1,\ldots,1}_{d-2\text{ times}}) & = \left.\ep^{-d+1} \sum_{\sigma\in S_{d-1}} \rb{\prod_{j=1}^{d-1} \sum_{i=1}^{j} x_{\sigma(i)}}^{-1} \right\vert_{x=(2,1,\ldots,1)}  +O(\ep^{-d+2})
\end{align}
If $x=(2,1,\ldots,1)$ then we have, for $j\in\set{1,\ldots,d-1}$ and $\sigma\in S_{d-1}$,
\begin{align}
	\sum_{i=1}^j x_{\sigma(i)} &= \begin{cases}
		j+1  \quad & \text{if } j<\sigma^{-1}(1)\\
		j & \text{otherwise,}
	\end{cases}
	\intertext{and therefore}
	\left.\sum_{\sigma\in S_{d-1}} \rb{\prod_{j=1}^{d-1}\sum_{i=1}^{j} x_{\sigma(i)}}^{-1} \right\vert_{x=(2,1,\ldots,1)}  & =  \sum_{\sigma\in S_{d-1}} \rb{\prod_{j=1}^{\sigma^{-1}(1)-1} j}^{-1}\rb{\prod_{j=\sigma^{-1}(1)}^{d-1} (j+1) }^{-1} \\
	&= \sum_{\sigma\in S_{d-1}}\frac{\sigma^{-1}(1)}{d!} =\frac{1}{d!} \sum_{r=1}^{d-1} \sum_{\sigma^{-1}(1)=r}\ r \\
	&= \frac{(d-2)!}{d!}\cdot \frac{d(d-1)}2=\frac12
	\end{align}	
	It follows that
	\begin{align}
		\label{eq:418}
		\Phi_{E'(e^{-\ep})} (2,\underbrace{1,\ldots,1}_{d-2\text{ times}}) & = \frac12 \ep^{-d+1} + O(\ep^{-d+2}).
	\end{align}
	Substituting \eqref{eq:413} and \eqref{eq:418} into \eqref{eq:first} and \eqref{eq:second} gives
	\begin{align}
		\sum_{\fM_d^{(n)}} f\rb{\mu^{(1)},\ldots,\mu^{(k)}}& W_{E'(q)}\rb{\mu^{(1)},\ldots,\mu^{(k)}} 
		=	\frac{ \ep^{-d}}{d!}\rb{1-\ep \frac{d(d+1)}4} f\rb{\underbrace{\bm2,\ldots,\bm2}_{d\text{ times}}}\\
		& \quad + \frac{ \ep^{-d+1} }{(d-1)!}\rb{f(\bm 3,\underbrace{\bm 2,\ldots,\bm 2}_{d-2\text{ times}}) + f(\bm{2}^2,\underbrace{\bm 2,\ldots,\bm 2}_{d-2\text{ times}})}	+ O(\ep^{-d+2})	
	\end{align}
as required.%
\end{proof}

Choosing $f\rb{\mu^{(1)},\ldots,\mu^{(k)}} = H\rb{\mu^{(1)},\ldots,\mu^{(k)},\mu,\nu}$ gives Theorem \ref{thm:WeightedHurwitz}. On the other hand by setting $f\rb{\mu^{(1)},\ldots,\mu^{(k)}}=1$ we obtain
\begin{align}
	\pf d{e^{-\ep}} &= \frac{\ep^{-d}}{d!} + \ep^{1-d}\rb{\frac1{(d-1)!} -\frac{d(d+1)}{4d!}}\\
	&= \frac{\ep^{-d}}{d!} + \ep^{1-d}\, \frac{3-d}{4(d-1)!} + O\rb{\ep^{2-d}}
\end{align}
and we have proved Proposition \ref{cor:pfSC}.

\ \\

\noindent 
\small{ {\it Acknowledgements.}  The authors would  like to thank G. Borot and A.Yu.~Orlov for helpful discussions.
The work of JH was partially supported by the Natural Sciences and Engineering Research Council of Canada (NSERC) and the Fonds de recherche du Qu\'ebec, Nature et technologies (FRQNT). JO was partially supported by a CRM-ISM postdoctoral fellowship.}

\newcommand{\arxiv}[1]{\href{http://arxiv.org/abs/#1}{arXiv:{#1}}}

\bigskip
\noindent


\begin{thebibliography}{99}

\bibitem{AC1} J. Ambj{\o}rn and L. Chekhov, ``The matrix model for dessins d'enfants'', 
{\em Ann. Inst. Henri Poincar\'e, Comb. Phys. Interact.} {\bf  1},  337-361 (2014).

\bibitem{AC2} J. Ambj{\o}rn and L. Chekhov, ``A matrix model for hypergeometric Hurwitz numbers'', 
{\em Theor. Math. Phys. } {\bf   181},  1486-1498 (2014).

\bibitem{AMMN1} 
  A.~Alexandrov, A.~Mironov, A.~Morozov and S.~Natanzon,
  ``Integrability of Hurwitz Partition Functions. I. Summary,''
  {\em J.\ Phys.\ A}  {\bf 45}, 045209 (2012).
  
  \bibitem{AMMN2}  A. Alexandrov, A. Mironov, A. Morozov  and S. Natanzon, ``
On KP-integrable Hurwitz functions'',  arXiv:1405.1395; {\em JHEP} doi:10.1007/JHEP11(2014)080.

\bibitem{Frob1}  G. Frobenius, ``\"Uber die Charaktere der symmetrischen Gruppe'', {\em Sitzber.  Akad. Wiss., Berlin},
516-534  (1900). Gesammelte Abhandlung III, 148-166.

\bibitem{Frob2}  G. Frobenius, ``\"Uber die Charakterische Einheiten der symmetrischen Gruppe'', {\em Sitzber.  Akad. Wiss., Berlin}, 328-358  (1903). Gesammelte Abhandlung III, 244-274.

 \bibitem{GGN} I. P. Goulden, M. Guay-Paquet and J. Novak, ``Monotone Hurwitz numbers and the HCIZ Integral'', 
 {\em Ann. Math. Blaise Pascal} {\bf 21} 71-99 (2014).

\bibitem{GH1} M. Guay-Paquet and J. Harnad, ``2D Toda $\tau$-functions as combinatorial generating functions'', 
{\em Lett. Math. Phys.} {\bf 105} Page 827-852 (2015).

\bibitem{GH2} M. Guay-Paquet and J. Harnad, ``Generating functions for weighted Hurwitz numbers'', 
 \arxiv{1408.6766}.

 \bibitem{H1} J. Harnad,  ``Weighted Hurwitz numbers and hypergeometric $\tau$-functions: an overview'', {\em  Proc.
Symp. Pure Math.} {\bf 93}: 289-333 (2016). \arxiv{1504.03408 }.
 
 \bibitem{H2} J. Harnad,  ``Quantum Hurwitz numbers  and Macdonald polynomials'', {\em J. Math. Phys.} (in press, 2016). \arxiv{1504.03311}.

\bibitem{HO} J. Harnad and A. Yu. Orlov, ``Hypergeometric $\tau$-functions, Hurwitz numbers and enumeration of paths'',  {\em Commun. Math. Phys. } {\bf 338}, 267-284 (2015).

\bibitem{Hu1} A. Hurwitz, ``\"Uber Riemann'sche Fl\"achen mit gegebenen Verzweigungspunkten'', 
{\em Math. Ann.} {\bf 39}, 1-61 (1891); Mathematische Werke I, 321-384.

\bibitem{Hu2} A. Hurwitz, ``\"Uber die Anzahl der Riemann'sche Fl\"achen mit gegebenen Verzweigungspunkten'', 
{\em Math. Ann.} {\bf 55}, 53-66 (1902); Mathematische Werke I, 42-505.

\bibitem{KZ} M. Kazarian and P. Zograf, ``Virasoro constraints and topological recursion for Grothendieck's dessin counting'', 
{\it Lett. Math. Phys.}  {\bf 105} 1057-1084 (2015).

\bibitem{KMMM}  S. Kharchev, A. Marshakov, A. Mironov and A. Morozov,
  ``Generalized Kazakov-Migdal-Kontsevich Model: group theory aspects'', 
{\em  Int. J. Mod. Phys.} {\bf A10} 2015 (1995).

 \bibitem{LZ} S. K.  Lando and A.K.  Zvonkin {\em Graphs on Surfaces and their Applications}, Encyclopaedia of Mathematical Sciences, Volume {\bf 141}, with  appendix by D. Zagier, Springer, N.Y. (2004).
 
\bibitem{Mac} I.~G. ~Macdonald, {\em Symmetric Functions and Hall Polynomials},
 Clarendon Press, Oxford, (1995).

\bibitem{Ok} A. Okounkov, ``Toda equations for Hurwitz numbers'', {\em Math.~Res.~Lett.} {\bf 7}, 447--453 (2000).

\bibitem{OrSc} A. Yu. Orlov and D. M. Scherbin, ``Hypergeometric solutions of soliton equations'', {\em Theoretical and Mathematical Physics} {\bf 128}, 906-926 (2001).

 \bibitem{Pa} R. Pandharipande, ``The Toda Equations and the Gromov-Witten Theory of
the Riemann Sphere'',  {\em Lett. Math. Phys.} {\bf  53}, 59-74 (2000).

\bibitem{Sch} I. Schur ``Neue Begr\"undung' der Theorie der Gruppencharaktere'',
 {\em Sitzber.  Akad. Wiss., Berlin}, 406-432  (1905).
 
 \bibitem{Z} P. Zograf, ``Enumeration of Grothendieck's dessins and KP hierarchy'', 
{\em Int. Math. Res. Notices} {\bf 2015}, No. 24, 13533-13544 (2015).


\end{thebibliography}
\end{document}